\documentclass[12pt]{amsart}

\usepackage{amsmath,amsthm,amsfonts,amssymb,latexsym,verbatim,bbm,hyperref}
\usepackage[shortlabels]{enumitem}
\usepackage[top=1.2in,bottom=1.2in,left=1.2in,right=1.2in]{geometry}
\usepackage{subcaption}
\usepackage{afterpage}
\usepackage{cleveref}

\DeclareRobustCommand{\SkipTocEntry}[5]{}

\allowdisplaybreaks

\setlist{itemsep=.5\baselineskip,topsep=.5\baselineskip}

\numberwithin{equation}{section}
\theoremstyle{plain}
\newtheorem{theorem}{Theorem}[section]

\newtheorem{lemma}[theorem]{Lemma}

\newtheorem{prop}[theorem]{Proposition}
\newtheorem{example}[theorem]{Example}

\newtheorem{defn}[theorem]{Definition}

\newtheorem{cor}[theorem]{Corollary}

\newtheorem{problem}[theorem]{Problem}

\newcommand{\iso}{\cong}
\newcommand{\arr}{\rightarrow}

\newcommand{\R}{\mathbb{R}}
\newcommand{\C}{\mathbb{C}}
\newcommand{\Z}{\mathbb{Z}}

\newcommand{\Id}{\mathbbm{1}}

\newcommand{\mcF}{\mathcal{F}}

\newcommand{\mcP}{\mathcal{P}}

\usepackage{braket}
\newcommand{\ang}[1]{\langle #1 \rangle}

\usepackage{tikz}
\usetikzlibrary{positioning,calc,arrows.meta}
\usetikzlibrary{decorations.pathmorphing,decorations.pathreplacing,patterns,shapes.symbols}

\title{Operator solutions of linear systems and small cancellation}

\author[William Slofstra]{William Slofstra$^{1,2}$} 
\author[Lu-Ming Zhang]{Lu-Ming Zhang$^{3}$}

\address[1]{Institute for Quantum Computing, University of Waterloo, Canada}
\address[2]{Department of Pure Mathematics, University of Waterloo,
Canada}
\address[3]{Department of Mathematics, London School of Economic and Political Science, UK}

\begin{document}
\begin{abstract}
    We show that if a graph has minimum vertex degree at least $d$ and girth at
    least $g$, where $(d,g)$ is $(3,6)$ or $(4,4)$, then the incidence system
    of the graph has a (possibly infinite-dimensional) quantum solution over
    $\Z_p$ for every choice of vertex weights and integer $p \geq 2$. In
    particular, there are linear systems over $\Z_p$, for $p$ an odd prime,
    such that the corresponding linear system nonlocal game has a perfect
    commuting-operator strategy, but no perfect classical strategy.
\end{abstract}

\maketitle

\section{Introduction}

Suppose $Ax=b$ is a linear system over $\Z_p = \Z / p \Z$, where $p \geq 2$ is
an integer (not necessarily prime). In the linear system nonlocal game
based on $Ax=b$, two players (commonly called Alice and Bob) are
each given an equation from the system chosen uniformly at random, and must
reply with an assignment to the variables which appear with non-zero
coefficients in their given equation.  The players win the game if their
assignments satisfy the given equations, and agree on the variables which
appear in both equations. The players are not able to communicate during the
game, and do not know what equation the other player received. However, they
can meet in advance and decide on a strategy which maximizes their winning
probability. It's not hard to see that Alice and Bob can play perfectly (i.e.
win with probability one) with a deterministic strategy if and only if the
system $Ax=b$ has a solution. The same is true if Alice and Bob use shared
randomness. In quantum information, we are interested in the maximal winning
probability when Alice and Bob have access to entanglement. Cleve and Mittal
\cite{CM14} show that Alice and Bob can play perfectly with a
finite-dimensional entangled quantum state if and only if $Ax=b$ has a
finite-dimensional operator solution, in the following sense:
\begin{defn}
    Let $Ax=b$ be an $m \times n$ linear system over $\Z_p$, $2 \leq p < +\infty$, and let
    $\omega_p := e^{2\pi i / p}$. An \textbf{operator solution} to $Ax=b$ is a
    collection of unitary operators $X_1,\ldots,X_n$ on a Hilbert space $H$
    such that
    \begin{enumerate}[(1)]
        \item $X_i^{p} = \Id$ for all $i=1,\ldots,n$,
        \item $\prod_j X_j^{A_{ij}} = \omega_p^{b_i} \Id$ for all $1 \leq i \leq m$, and
        \item $[X_j, X_k] = \Id$ for all $1 \leq i \leq m$ and $1 \leq j,k \leq n$
            with $A_{ij}$, $A_{ik} \neq 0$,
    \end{enumerate}
    where $[A,B] := A B A^{-1} B^{-1}$ is the group commutator.  An operator
    solution is \textbf{finite-dimensional} if $H$ is finite-dimensional.
\end{defn}
For the commuting operator model (a more general model for entangled systems
based on the maximal tensor product of $C^*$-algebras), Cleve, Liu, and the
first author show that the linear system game associated to $Ax=b$
has a perfect commuting operator strategy if and only if $Ax=b$ has an operator
solution (possibly infinite-dimensional) \cite{CLS}. While this result in
\cite{CLS} is claimed for all integers $p$, the details are given only for
$p=2$. A complete development for all $p$ can be found in Qassim-Wallman
\cite{QW}. Both approaches use the solution group of the system $Ax=b$, also
introduced in \cite{CLS}:
\begin{defn}\label{def:solngroup}
    Suppose $A$ is an $m \times n$ integer matrix, $b \in \Z^m$, and $2 \leq p \leq +\infty$
    is an integer. The \textbf{solution group} of $Ax=b$ over $\Z_p$ is the finitely
    presented group $\Gamma_p(A,b)$ generated by $x_1,\ldots,x_n$, $J$,
    satisfying the relations
    \begin{enumerate}[(1)]
        \item\label{I:orderrels} $J^p = 1$ and $x_i^p=1$ for all $1 \leq i \leq n$, 
        \item $\prod_{j=1}^n x_j^{A_{ij}} = J^{b_i}$ for all $1 \leq i \leq m$, 
        \item $[x_j,x_k] = 1$ for all $1 \leq i \leq m$ and $1 \leq j,k \leq n$
            with $A_{ij}$, $A_{ik} \neq 0$, and
        \item $[x_j, J]=1$ for all $1 \leq i \leq n$.
    \end{enumerate}
    When $p = \infty$, the relations in (1) are omitted, and we let $\Z_{\infty} := \Z$.
\end{defn}
Note that there is no loss of generality in assuming that $A$ and $b$ have
integer entries, since we can always pick integer representatives for elements
of $\Z_p$. Using integer entries in this definition allows us to talk about
$\Gamma_p(A,b)$ for different values of $p$ with $A$ and $b$ fixed, including
$p = \infty$. For $2 \leq p < +\infty$, an operator solution of $Ax=b$ over
$\Z_p$ is a representation of $\Gamma_{p}(A,b)$ where $J \mapsto \omega_p \Id$.
Consequently $Ax=b$ has an operator solution over $\Z_p$ if and only if $J$ has
order $p$ in $\Gamma_p(A,b)$. We are not aware of a game interpretation of
$\Gamma_{\infty}(A,b)$, but it is still an interesting algebraic object. 

When $p=2$, the Mermin-Peres magic square and pentagram are famous examples of
linear systems which have no classical solution, but do have a
finite-dimensional operator solution \cite{Me90,Peres91}. There are also examples for $p=2$ of
linear systems which have no finite-dimensional operator solution, but do have
infinite-dimensional operator solutions \cite{Sl20,Sl19}. It's natural to
ask whether there are similar examples for $p > 2$, but it's been surprisingly
hard to find such examples. Qassim and Wallman prove that if $p > 2$ and $Ax=b$
does not have a classical solution, then $Ax=b$ does not have an operator
solution in the generalized Pauli group \cite{QW}. Similar results for
extraspecial $p$-groups and other specific finite groups have been given by
Chung, Okay, and Sikora \cite{COS} and Frembs, Okay, and Chung \cite{FOC}.

The main result of this paper is a sufficient condition for showing that $Ax=b$
has an operator solution over $\Z_p$: 
\begin{theorem}\label{T:main}
    Suppose $2 \leq p \leq +\infty$, and let $A$ be an $m \times n$ integer matrix
    such that every non-zero entry of $A$ is not a zero divisor in $\Z_p$.
    Let $H(A)$ be the hypergraph with vertex set $[m] := \{1,\ldots,m\}$, edge
    set $[n]$, and vertex $i$ incident to edge $j$ if $A_{ij} \neq 0$. If
    $H(A)$ has minimum vertex degree $\geq d$ and girth $\geq g$, where $(d,g)
    \in \{(4,4),(3,6)\}$, then $|J| = p$ in $\Gamma_p(A,b)$ for all $b \in
    \Z^m$.
\end{theorem}
In this theorem, a hypergraph is a triple $(V,E,I)$, where $V$ and $E$ are
finite sets called the vertex and edge sets, and $I$ is a binary relation
between $V$ and $E$ called the incidence relation. The degree of a vertex $v$
is the number of edges incident to $v$. A Berge cycle of length $k$ in $H$ is a
sequence $(v_1,e_1,\ldots,v_k,e_k,v_{k+1})$ such that $k \geq 2$, $v_1,\ldots,v_k$
is a sequence of distinct vertices, $v_{k+1}=v_1$, $e_1,\ldots,e_k$ is a
sequence of distinct edges, and $e_i$ is incident to $v_i$ and $v_{i+1}$ for
all $1 \leq i \leq k$. If $H$ is an ordinary graph, then a Berge cycle is the
same as an ordinary cycle. The girth of a hypergraph is the minimum length of a
Berge cycle, assuming one exists. If no cycle exists, then the girth is
infinite. If the entries of $A$ are zero or $\pm 1$, then the requirement that
the non-zero entries be non-zero divisors is satisfied for all $p$.

From Theorem \ref{T:main}, it is not hard to find examples of linear systems
$Ax=b$ which have operator solutions over $\Z_p$, but no classical solutions.
\begin{example}\label{ex:hypergraph}
    A hypergraph is $d$-regular if every vertex is incident to $d$ edges, and
    $p$-uniform if every edge is incident to $p$ vertices. Let $H$ be a 
    $d$-regular $p$-uniform hypergraph of girth $\geq g$ with vertex set $V
    = [m]$ and edge set $E = [n]$, where $d \geq d_0$
    for some $(d_0,g) \in \{(4,4), (3,6)\}$. Such hypergraphs exist by \cite{EL14}. 
    Let $A$ be the $m \times n$ matrix where $A_{ij} = 1$ if vertex $i$ is
    incident to edge $j$, and $A_{ij} = 0$ otherwise. Let $b \in \Z^m$ be any
    vector such that $\sum_i b_i$ is non-zero in $\Z_p$. By \Cref{T:main},
    $Ax=b$ has an operator solution over $\Z_p$. However, since $H$ is
    $p$-uniform, every column of $A$ has exactly $p$ entries, so if
    $\overline{e}$ is the vector of all ones, then $\overline{e}^T \cdot A = 0$
    in $\Z_p$. Hence $Ax=b$ does not have a solution in $\Z_p^n$. 
\end{example}

\begin{example}\label{ex:incidence}
    Let $(d,g) \in \{(4,4), (3,6)\}$, and let $G$ be a graph of minimum vertex
    degree $\geq d$ and girth $\geq g$ with vertex set $[m]$ and edge set
    $[n]$. Turn $G$ into a digraph $D$ by picking directions for the edges of
    $G$ arbitrarily (see Lemma \ref{lem:changeorientation} for why the choice
    of edge directions doesn't matter). The incidence
    matrix of $D$ is the $m \times n$ matrix $I(D)$ defined by 
    \begin{equation*}
        I(D)_{ij} = \begin{cases} 1 & t(j) = i \\
                        -1 & s(j) = i \\
                        0 & \text{otherwise} \end{cases},
    \end{equation*}
    where $s(j)$ is the source of edge $j$, and $t(j)$ is the target of edge $j$.
    If $\overline{e}$ is the vector of all ones, 
    then $\overline{e}^T I(D) = 0$ in $\Z$. So if $I(D)x=b$ has a solution in $\Z_p$, 
    then $\overline{e}^T b = \sum_i b_i = 0$ in $\Z_p$. In fact, it is not hard
    to see that if $G$ is connected, then $I(D) x = b$ has a solution in $\Z_p$
    if and only if $\sum_i b_i = 0$ in $\Z_p$. In contrast, $H(I(D)) = G$, so
    by \Cref{T:main}, $|J|=p$ in $\Gamma_p(I(D),b)$ for all $2 \leq p \leq +\infty$
    and $b \in \Z_p^n$.
\end{example}
Examples \ref{ex:hypergraph} and \ref{ex:incidence} show that there are linear
system games over $\Z_p$ for $p > 2$ which have perfect commuting operator
strategies but no perfect classical strategies. We do not know if these
examples have finite-dimensional operator solutions. Hence it remains an open
problem to determine if there are linear system games over $\Z_p$ for $p > 2$
which have perfect finite-dimensional quantum strategies, but does not have a
perfect classical strategy.  The following simple example shows why we require
the non-zero entries of $A$ to be non-zero divisors in \Cref{T:main}:
\begin{example}\label{ex:zerodivisor}
    Continue with the digraph $D$ from the previous example, and consider the
    system $2 I(D) x = b$ over $\Z_4$, for any vector $b$ such that $2b \neq 0$
    in $\Z_4$. It's easy to see that this system has no classical solutions,
    since if $x_0$ is a solution, then multiplying both sides of the equation by $2$ 
    gives $0 = 4 I(D)x_0 = 2b \neq 0$, a contradiction. Similarly, by squaring
    both sides of the relations (2) in \Cref{def:solngroup}, we see that $J^2=1$
    in $\Gamma_4(2I(D),b)$. So the conclusion of \Cref{T:main} does not hold, 
    even though $H(2I(D)) = G$. 
\end{example}

The proof of \Cref{T:main} is given in \Cref{sec:pictures}, and is based on
combinatorial small cancellation theory. In small cancellation theory, a group
presentation $K = \ang{S : R}$ is said to be symmetrized if every element of
$R$ is cyclically reduced, and $R$ is closed under inverses and cyclic
reorderings. A \textbf{piece} of a relation is a word $u \in \mcF(S)$ which
appears as an initial subword of two distinct elements of $R$. A symmetrized
presentation is said to satisfy the $C(d)$ condition if every element of $R$
cannot be written as $u_1 \cdots u_k$ for pieces $u_1,\ldots,u_k$ for any $k <
d$. A symmetrized presentation is said to satisfy the $T(g)$ condition if for
any sequence $r_0,\ldots,r_{k-1},r_k=r_0 \in R$ with $3 \leq k < g$, 
either $r_i r_{i+1}$ is reduced or $r_i = r_{i+1}^{-1}$ for some $0 \leq i \leq
k-1$. The point of these conditions is that if $C(d)$ holds,
then every vertex in a \textbf{picture} of $K$ has degree $\geq d$, and if
$T(g)$ holds, then every face has degree $\geq g$ (we refer to \cite{LS77,Sh07}
for more background on the terminology used here). If one of the hypotheses
$C(6)$, $C(4)$ and $T(4)$, or $C(3)$ and $T(6)$ holds, then small cancellation
theory can be applied to, for instance, solve the word problem of $K$.  

Unfortunately, even if we ignore the relations in (1) and (4), the presence of
the commuting relations in (3) means that none of these small cancellation
hypotheses hold for the presentation of $\Gamma_p(A,b)$ in
\Cref{def:solngroup}. We can also try replacing the relations (2) and (3) with
the equivalent relations $\prod_{j=1}^n x_{\sigma(j)}^{A_{i\sigma(j)}} =
J^{b_i}$ as $\sigma$ ranges across all permutations of $n$, but then any
subword of one of these relations is a piece, and when the relations (4) are
included, none of the small cancellation hypotheses will hold.
The main observation of this paper is that we can get around these problems by
using a variant of the pictures for solution groups introduced in \cite{Sl20}.
This allows us to conclude that $|J| = p$ in $\Gamma_p(A,b)$ when the
appropriate analogs of the $C(4)$---$T(4)$ or $C(3)$---$T(6)$ small
cancellation conditions are satisfied.  We do not know if there is an analog of
$C(6)$ small cancellation in this context. Although we give an elementary proof
of \Cref{T:main}, we suspect that it is also possible to prove this theorem
using complexes of groups as in \cite{BH}. We leave this for future work.

Among linear system games, we are particularly interested in the incidence
systems of graphs $G$ described in \Cref{ex:incidence}. For instance, the
Mermin-Peres magic square is an incidence system of the complete bipartite
graph $K_{3,3}$, and the magic pentagram is an incidence system for the complete
graph $K_5$. When $p=2$, Arkhipov's theorem states that there is a vector $b$
such that $I(G)x=b$ has an operator solution but no classical solution if and
only if $G$ is nonplanar, or in other words if $G$ contains $K_{3,3}$ or $K_5$
as a graph minor. Paddock, Russo, Silverthorne, and the first author
\cite{PRSS} show that any quotient closed property of solution groups can be
characterized by forbidden minors. The results are stated for $p=2$, but the
proofs can be easily extended to $\Gamma_p(I(G),b)$ for $2 \leq p \leq +\infty$.
In particular, the condition that $|J|=p$ in $\Gamma_p(I(G),b)$ can be
characterized by forbidden minors. Finding the minors explicitly seems
like an interesting open problem. We make some remarks on this in
\Cref{sec:incidence}. 

\subsection{Acknowledgements}

We thank Josse van Dobben de Bruyn, Cihan Okay, Denis Osin, David Roberson, and
Turner Silverthorne for helpful conversations and comments.  This work was
supported by NSERC DG 2018-03968 and an Alfred P. Sloan Research Fellowship.

\section{Proof of \Cref{T:main}}\label{sec:pictures}
In this section we define a version of closed pictures (meaning pictures without
boundary) for solution groups,
and use them to prove \Cref{T:main}. A picture is a plane digraph, or in other
words a directed graph embedded in the plane $\R^2$. If $e$ is an edge in a
digraph, we let $s(e)$ denote the source vertex, and $t(e)$ denote the target
vertex. We say that $v$ and $e$ are incident if $s(e) = v$ or $t(e) = v$.
Digraphs are allowed to have loops, which are edges $e$ with $s(e) = t(e)$. It
is often convenient when defining pictures to allow closed loops, which are
edges which aren't incident to any vertex. However, these are not necessary for
the argument below. We continue with the notation from the introduction that
$[n] := \{1,\ldots,n\}$ and $\Z_{\infty} := \Z$. 
\begin{defn}
    Suppose $Ax=b$ is an $m \times n$ linear system over $\Z$, and 
    $2 \leq p \leq +\infty$. For all $i \in [m]$, let 
    $R_i(A) = \{j \in [n] : A_{ij} \neq 0\}$, and let 
    \begin{align*}
         \Gamma_p(A,b,i) := {\Big\langle} J, x_j, j \in R_i(A)\ :\  
                & \prod_{j \in R_i(A)} x_j^{A_{ij}} = J^{b_i}, \\
                & J^p = x_j^p = [x_j,J] = [x_j,x_k] = 1 \text{ for all } j,k \in R_i(A) 
                {\Big\rangle},
    \end{align*}
    where we drop the relations $J^p = x_j^p = 1$ if $p = \infty$. A
    \textbf{closed picture} for the solution group $\Gamma_p(A,b)$ is a plane
    digraph $\mcP = (V,E)$, along with functions $h_E : E \arr [n]$, $h_V : V
    \arr [m]$, $a : E \arr \Z_p$, and $k : V \arr \Z_p$, such that
    \begin{enumerate}[(a)]
        \item if $e$ is incident to $v$ then $h_E(e) \in R_{h_V(v)}(A)$, and
        \item for all vertices $v \in V$,
            \begin{equation}\label{eq:vertex}
                \prod_{\substack{e \in E \\ t(e) = v}} x_{h_E(e)}^{a(e)} \cdot 
                    \prod_{\substack{e \in E \\ s(e) = v}} x_{h_{E}(e)}^{-a(e)} = J^{k(v)}
            \end{equation}
            in $\Gamma_p(A,b,{h_V(v)})$. 
    \end{enumerate}
    The \textbf{phase} of a picture $\mcP$ is $k = \sum_{v \in V} k(v) \in \Z_p$. A picture 
    is \textbf{minimal} if there is no other picture with the same phase but fewer vertices and edges. 
\end{defn}

We refer to the functions $h_E$, $h_V$, $a$, and $k$ as the labels of a picture
$\mcP$. For convenience, we write $h$ for $h_V$ and $h_E$, since we can tell
which one is meant by whether the argument is a vertex or edge. The point of
pictures is:
\begin{lemma}[van Kampen lemma]
    Let $Ax=b$ be a linear system over $\Z$ and let $2 \leq p \leq +\infty$. If
    $k \in \Z_p$, then $J^k = 1$ in $\Gamma_p(A,b)$ if and only if there is a
    closed picture of phase $k$.
\end{lemma}
We could also define pictures with boundary if we wanted to capture other
relations in the group, but closed pictures are sufficient for the proof of
\Cref{T:main}. Although the edge orientations are an important part of the data
of a $\Gamma(A,b)$ picture, we can change the orientations arbitrarily as long
as we also change the edge labels accordingly:
\begin{lemma}\label{lem:orientation}
    Let $S$ be a subset of the edges of a picture $\mcP$ of $\Gamma(A,b)$.
    Define a new picture $\mcP'$ from $\mcP$ by changing, for all $e \in S$,
    the orientation of the edge $e$ and the label of $a(e)$ of $e$ to $-a(e)$, and
    leaving all other edges and labels unchanged. Then $\mcP'$ is also a $\Gamma(A,b)$
    picture with the same phase as $\mcP$.
\end{lemma}

\begin{figure}[t]
    \begin{tikzpicture}[vertex/.style={circle,draw,thin,inner sep=2.5}, every node/.style={scale=.8}]

    \begin{scope}
        \node[vertex] (1) {1};
        \node[vertex] (2) [right=of 1] {2};
        \node[vertex] (3) [right=of 2] {3};
        \node[vertex] (a) [below=of 1] {a};
        \node[vertex] (b) [below=of 2] {b};
        \node[vertex] (c) [below=of 3] {c};

        \draw[-Stealth] (a) to (1);
        \draw[-Stealth] (a) to (2);
        \draw[-Stealth] (a) to (3);
        \draw[-Stealth] (b) to (1);
        \draw[-Stealth] (b) to (2);
        \draw[-Stealth] (b) to (3);
        \draw[-Stealth] (c) to (1);
        \draw[-Stealth] (c) to (2);
        \draw[-Stealth] (c) to (3);
    \end{scope}

    \begin{scope}[shift={++(7,-.5)}]
        \node[vertex] (1o) at (0:2) {1};
        \node[vertex] (ao) at (60:2) {a};
        \node[vertex] (2o) at (120:2) {2};
        \node[vertex] (bo) at (180:2) {b};
        \node[vertex] (3o) at (240:2) {3};
        \node[vertex] (co) at (300:2) {c};

        \node[vertex] (1i) at (180:1) {1};
        \node[vertex] (ai) at (240:1) {a};
        \node[vertex] (2i) at (300:1) {2};
        \node[vertex] (bi) at (0:1) {b};
        \node[vertex] (3i) at (60:1) {3};
        \node[vertex] (ci) at (120:1) {c};
        \draw[-Stealth] (ao) to (2o);
        \draw[-Stealth] (ao) to (1o);
        \draw[-Stealth] (bo) to (2o);
        \draw[-Stealth] (bo) to (3o);
        \draw[-Stealth] (co) to (3o);
        \draw[-Stealth] (co) to (1o);
        \draw[-Stealth] (ai) to (2i);
        \draw[-Stealth] (ai) to (1i);
        \draw[-Stealth] (bi) to (2i);
        \draw[-Stealth] (bi) to (3i);
        \draw[-Stealth] (ci) to (3i);
        \draw[-Stealth] (ci) to (1i);

        \draw[-Stealth] (ao) to (3i);
        \draw[-Stealth] (bo) to (1i);
        \draw[-Stealth] (co) to (2i);
        \draw[-Stealth] (ai) to (3o);
        \draw[-Stealth] (bi) to (1o);
        \draw[-Stealth] (ci) to (2o);
    \end{scope}
    \end{tikzpicture}
    \caption{The graph $K_{3,3}$ with edges oriented from top to bottom, and a
        picture for the incidence system of $K_{3,3}$ corresponding to the
        double cover of $K_{3,3}$.} \label{fig:K33picture}
\end{figure}
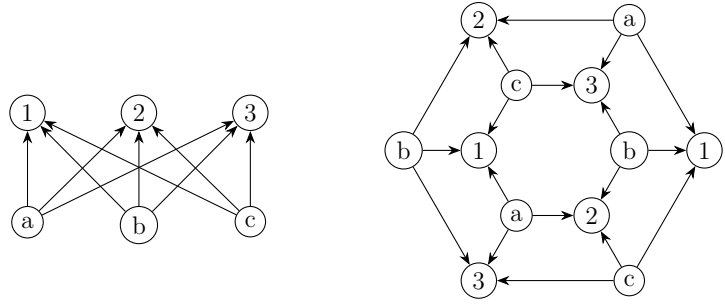

The following example shows how the van Kampen lemma can be used:
\begin{example}\label{ex:K33}
    Suppose $2 \leq p \leq +\infty$, and consider the graph $K_{3,3}$ shown on
    the left in \Cref{fig:K33picture}, with edges oriented from bottom to top 
    (this choice of orientation is arbitrary, but convenient).
    In \Cref{ex:incidence}, we assumed that the graph had vertex set $[m]$ and
    edge set $[n]$, but in examples it's convenient to allow other vertex and
    edge sets, and index the rows and columns of $I(D)$ by the vertices and
    edges respectively. In this case, we use vertex set $W = \{1,2,3,a,b,c\}$. 
    The Mermin-Peres magic square is the incidence system $I(K_{3,3})x = b$,
    where $b \in \Z^W$ is a vector with $\sum_{w \in W} b(w) = 1$.  Mermin and
    Peres originally showed that $J$ is non-trivial in $\Gamma_2(I(K_{3,3}),b)$.
    Qassim and Wallman show that $|J|<p$ in $\Gamma_p(I(K_{3,3}),b)$ for all
    $2 < p \leq +\infty$ \cite{QW}. This can be seen using the picture $\mcP$
    shown on the right of \Cref{fig:K33picture}.  In this picture, the
    underlying digraph is the double cover of $K_{3,3}$, with edges oriented to
    match the orientation from $K_{3,3}$. The functions $h_V$ and $h_E$ for
    this picture take values in the vertex and edge sets of $K_{3,3}$ respectively,
    since the row and columns of $I(D)$ are indexed by these sets. The function
    $h_V$ is indicated by the vertex labels shown in \Cref{fig:K33picture}, so
    if $v$ is a vertex of $\mcP$ with label $w \in W$, then $h_V(v) = w$.
    The function $h_E$ is described implicitly by the vertex
    labels, by letting $h_E(e)$ be the edge of $K_{3,3}$ with endpoints
    $h_V(s(e))$ and $h_V(t(e))$.  Finally, we let $a(e) = 1$ for all edges $e$
    and $k(v) = b(h_V(v))$ for all vertices $v$. With these labellings,
    \Cref{eq:vertex} is satisfied for all vertices, and $\mcP$ is a picture.
    Because every vertex of $K_{3,3}$ appears twice in $\mcP$, the phase of this
    picture is $2 \sum_{w \in W} b(w) = 2$, which implies that $J^2 = 1$ in
    $\Gamma_p(I(K_{3,3},b))$. Hence $|J|<p$ if $2 < p \leq +\infty$.
\end{example}

The faces of a plane digraph $\mcP$ are the connected components of $\R^2
\setminus \mcP$. By a simple undirected cycle in a digraph, we mean a sequence
$(v_1,e_1,\ldots,v_k,e_k,v_{k+1})$, where $k \geq 2$, $v_1,\ldots,v_k$ is a
sequence of distinct vertices, $v_{k+1} = v_1$, $e_1,\ldots,e_k$ is a sequence
of distinct edges, and $e_i$ is incident to $v_i$ and $v_{i+1}$ for all $1 \leq
i \leq k$. 
\begin{lemma}\label{lem:minimal}
    Let $2 \leq p \leq +\infty$, and let $Ax=b$ be an $m \times n$ linear system over $\Z$
    such that every non-zero entry of $A$ is not a zero divisor in $\Z_p$.
    If $\mcP$ is a minimal picture of $\Gamma_p(A,b)$, then:
    \begin{enumerate}[(a)]
        \item $\mcP$ has no loops, and no edges with $a(e)=0$.
        \item If $v,w$ are distinct vertices of $\mcP$ contained in the boundary of a common face, then $h(v) \neq h(w)$.  
        \item If $v$ is a vertex of $\mcP$, then $\deg(v) \geq \min(4,|R_{h(v)}(A)|)$.
        \item If $C$ is a simple undirected cycle in $\mcP$, then $C$ contains
            two edges $e_1$ and $e_2$ with $h(e_1) \neq h(e_2)$. 
    \end{enumerate}
\end{lemma}
\begin{proof}
    For part (a), if $s(e)=t(e)=v$, then $e$ contributes $x_{h(e)}^{a(e)} \cdot
    x_{h(e)}^{-a(e)} = 1$ to the product in \Cref{eq:vertex}, and
    hence $\mcP \setminus e$ will be a picture with the same phase. So if $\mcP$ is
    minimal, then $\mcP$ does not have loops. Similarly, if $\mcP$ has an edge $e$
    with $a(e)=0$, then we can remove $e$ without changing \Cref{eq:vertex} for any vertex $v$.

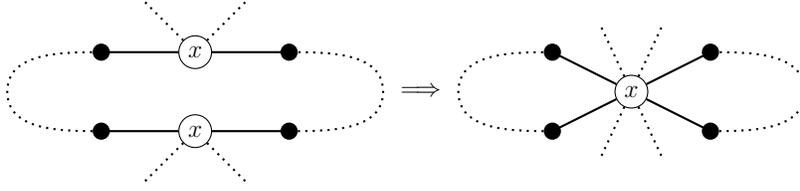
\begin{figure}[b]
    \begin{tikzpicture}[auto, thick, scale=.5, vertex/.style={circle,draw,thin,inner sep=2.5}, empty/.style={inner sep=0}, every node/.style={scale=.8}]

        \node[empty] (1) {};
        \node[vertex] (v) [right=of 1] {$x$}
            edge (1);
        \node[empty] (2) [right=of v] {}
            edge (v);
        \node[empty] (3) [below=of 1] {};
        \node[vertex] (w) [right=of 3] {$x$}
            edge (3);
        \node[empty] (4) [right=of w] {}
            edge (w);

        \draw [dotted] (v) to ($(v) + (45:2)$);
        \draw [dotted] (v) to ($(v) + (135:2)$);
        \draw [dotted] (w) to ($(w) + (-45:2)$);
        \draw [dotted] (w) to ($(w) + (-135:2)$);

        \draw [dotted] (1) to [out=180,in=90] ($.5*(1) + .5*(3) + (-2.5,0)$) to [out=-90,in=180] (3);
        \draw [dotted] (2) to [out=0,in=90] ($.5*(2) + .5*(4) + (+2.5,0)$) coordinate (r) to [out=-90,in=0] (4);

        \node (m) at ($(r) + (1,0)$) {$\implies$};
        \coordinate (l) at ($(r) + (2,0)$);

        \node[empty] (5) at ($(l) + (r) - (4)$) {};
        \node[empty] (ov) [right=of 5] {};
        \node[empty] (6) [right=of ov] {};
        \node[empty] (7) [below=of 5] {};
        \node[empty] (ow) [below=of ov] {};
        \node[empty] (8) [below=of 6] {};
        \node[vertex] (vw) at ($.5*(ov) + .5*(ow)$) {$x$}
            edge (5) edge (6) edge (7) edge (8);
        
        \draw [dotted] (vw) to ($(vw) + (65:2)$);
        \draw [dotted] (vw) to ($(vw) + (115:2)$);
        \draw [dotted] (vw) to ($(vw) + (-65:2)$);
        \draw [dotted] (vw) to ($(vw) + (-115:2)$);

        \draw [dotted] (5) to [out=180,in=90] (l) to [out=-90,in=180] (7);
        \draw [dotted] (6) to [out=0,in=90] ($.5*(6) + .5*(8) + (+2.5,0)$) to [out=-90,in=0] (8);

        \foreach \x in {1,...,8}
            \draw [fill] (\x) circle [radius=.2];
    
    \end{tikzpicture}
    \caption{If two vertices in a common face have the same label $x$, then we can contract the vertices
            to a single vertex to get a picture with the same phase but fewer vertices.}\label{fig:contractvertex}
\end{figure}

    For part (b), suppose $v$ and $w$ are contained in the boundary of a common
    face $F$, and $h(v) = h(w)$.  Contract $v$ and $w$ along a path through $F$
    to get a new planar graph $\mcP'$ where $v$ and $w$ are replaced with a new
    vertex $u$, all edges that were incident to $v$ and $w$ are now incident to
    $u$ (any edges between $v$ and $w$ turn into loops), and all other vertices
    and edges are unchanged, as shown in \Cref{fig:contractvertex}.
    If we label $u$ by $h(v)=h(w)$, set the phase of $u$ to be $k(v)+k(w)$, and
    leave everything else from $\mcP$ unchanged, then $\mcP'$ will be a picture
    of $\Gamma_p(A,b)$ with the same phase as $\mcP$, and one less vertex. 

\begin{figure}
    \begin{tikzpicture}[auto, thick, scale=.5,vertex/.style={circle,draw,thin,inner sep=2.5}, empty/.style={inner sep=0}, every node/.style={scale=.8}]

    \begin{scope}
        \node[empty] (v) {} node[below,shift={(-.1,-.15)}] {$v$};
        \draw[Stealth-] (v) to node[left] {$e_1$} (135:4) coordinate (1);
        \draw[Stealth-] (v) to node[below] {$e_2$} (0:4) coordinate (2);
        \draw[-Stealth,dashed] (1) to [bend left] node[below] {$e'$} node[above,shift={(.7,.2)}] {$a(e') = a(e_1)$} (2);

        \foreach \x in {v,1,2}
            \draw [fill] (\x) circle [radius=.2];
    \end{scope}

    \begin{scope}[shift={(15,0)}]
        \node[empty] (v) {} node[below,shift={(-.1,-.15)}] {$v$};
        \draw[Stealth-] (v) to node[left,pos=.4] {$e_3$} (90:3.5) coordinate (3);
        \draw[Stealth-] (v) to node[below] {$e_2$} (-20:3.5) coordinate (2);
        \draw[Stealth-] (v) to node[below] {$e_1$} (200:3.5) coordinate (1);

        \draw[-Stealth,dashed] (1) to node[right,pos=.4] {$f_1$} node[left,shift={(-.2,.3)}] {$a(f_1) = a(e_1)$} (3);
        \draw[-Stealth,dashed] (2) to node[left,pos=.4] {$f_2$} node[right,shift={(.2,.3)}] {$a(f_2) = a(e_2)$} (3);

        \foreach \x in {v,1,2,3}
            \draw [fill] (\x) circle [radius=.2];
    \end{scope}
        
    \end{tikzpicture}
    \caption{If $v$ has two or three incident edges with the same label, then we can delete $v$ and replace the incident
        edges with the dotted edges to get a smaller picture with the same phase.}\label{fig:smalldegremoval}
\end{figure}
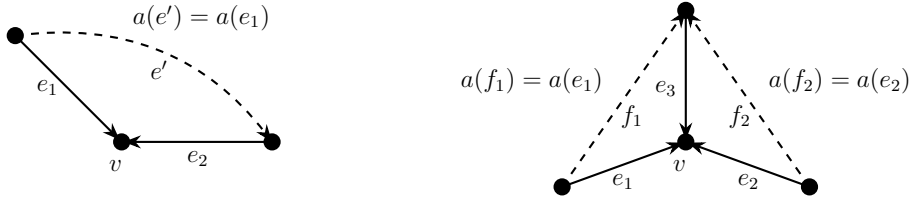

    For part (c), suppose $\deg(v) < \min(4,|R_{h(v)}(A)|)$, and let $R' =
    \{h(e) : e \text{ incident to } v\}$. By part (a), we can assume that
    none of the edges incident to $v$ are loops, and that $a(e) \neq 0$ for all
    edges $e$. Since $|R'| \leq \deg(v)$, $R'$ is a strict subset of
    $R_{h(v)}(A)$. Let $R_{h(v)}(A) = \{j_1,\ldots,j_\ell\}$ where $j_1 <
    \ldots < j_{\ell}$, let $\beta_i$ be the $i$th coordinate vector of
    $\Z_p^{\ell+1}$, and let 
        $\alpha := (A_{h(v)j_1},\ldots,A_{h(v)j_{\ell}}, - b_{h(v)}) \in \Z_{p}^{\ell+1}$.
    There is an isomorphism $\Gamma_p(A,b,h(v)) \to \Z_{p}^{\ell+1} / \ang{\alpha}$ sending $x_{j_i} \mapsto \beta_i$
    for all $1 \leq i \leq \ell$ and $J \mapsto \beta_{\ell+1}$.  With this
    presentation for $\Gamma_p(A,b,h(v))$, \Cref{eq:vertex} states that the
    vector
    \begin{equation*}
        \gamma := \sum_{i = 1}^{\ell} \left(\sum_{\substack{e \in R' : h(e) = j_i\\ \text{ and } s(e)=v }} a(e)
            - \sum_{\substack{e \in R': h(e) = j_i \\ \text{ and }t(e)=v}} a(e)\right) \beta_i - k(v) \beta_{\ell+1}
    \end{equation*}
    belongs to $\ang{\alpha}$, so $\gamma = r \alpha$ for some $0 \leq r < p$.
    Out of the first $\ell$ coordinates of $\gamma$, at most $|R'| < \ell$ are non-zero.
    But since $A_{h(v)j_i}$ is a non-zero divisor for all $1 \leq i \leq \ell$, 
    the first $\ell$ coordinates of $r \alpha$ are non-zero for all $0 < r < p$.
    We conclude that $\gamma = 0 \cdot \alpha = 0$. 
    In particular, $k(v) =0$, so if $\deg(v) = 0$, then $\mcP \setminus v$ is a
    picture with the same phase, contradicting the minimality of $\mcP$.
    We also see that for every $j_i \in R'$, there must be at least two edges
    $e \in R'$ such that $h(e) = j_i$, since otherwise the $i$th coordinate of
    $\gamma$ is $\pm a(e) \neq 0$. This implies that the degree of $v$ can't be one,
    and if $\deg(v) = 2$ or $3$, then $R'$ consists of a single element $j_1 = j$ 
    (this conclusion would not hold if $\deg(v)=4$). 
    Since none of the edges
    incident to $v$ are loops, we can apply \Cref{lem:orientation} to 
    assume that all edges $e$ incident to $v$ are incoming, or in other
    words have $t(e) = v$. If $v$ is incident to two edges $e_1$ and $e_2$,
    then $a(e_1) = - a(e_2)$. Hence we can delete $v$ and replace $e_1$ and
    $e_2$ with a single edge $e'$ from $s(e_1)$ to $s(e_2)$ with labels
    $h(e') = h(e_1)$ and $a(e') = a(e_1)$, as shown on the left in \Cref{fig:smalldegremoval}.
    This yields a picture with the same phase as $\mcP$, and one less vertex
    and edge, contradicting minimality again. If $v$ is incident to
    three edges $e_1$, $e_2$, and $e_3$, then $a(e_1) + a(e_2) + a(e_3)= 0$, and
    we can again delete $v$ and replace $e_1$, $e_2$, and $e_3$ with two edges
    $f_1$ and $f_2$ such that $s(f_1) = s(e_1)$, $s(f_2) = s(e_2)$, $t(f_1) =
    t(f_2) = s(e_3)$, $h(f_1) = h(f_2) = h(e_1)$, $a(f_1) = a(e_1)$, and
    $a(f_2) = a(e_2)$, as shown on the right in \Cref{fig:smalldegremoval}.
    Once again, we get a picture with the same phase as $\mcP$ but one less
    vertex and edge.

    \begin{figure}
        \begin{tikzpicture}[auto, thick, scale=.5,vertex/.style={circle,draw,thin,inner sep=2.5}, empty/.style={inner sep=0}, every node/.style={scale=.8}]

        \begin{scope}
            \node[empty] (1) {};
            \node[empty] (2) [right=of 1,shift={(.5,0)}] {};
            \node[empty] (3) [right=of 2,shift={(.5,0)}] {};
            \node[empty] (4) [right=of 3,shift={(.5,0)}] {};
            \node at ($(1)+(-.1,-.6)$) {$v_1$};
            \node at ($(2)+(-.1,-.6)$) {$v_2$};
            \node at ($(3)+(-.1,-.6)$) {$v_{k-1}$};
            \node at ($(4)+(-.1,-.6)$) {$v_k$};
            \draw[-Stealth] (1) to node[below] {$e_1$} (2);
            \draw[dashed] (2) to (3);
            \draw[-Stealth] (3) to node[below] {$e_{k-1}$} (4);
            \draw[-Stealth] (1) to[bend left] node[above] {$e_k$} (4);

            \node [right=of 4] {$\implies$};

            \foreach \x in {1,2,3,4}
                \draw [fill] (\x) circle [radius=.2];
        \end{scope}

        \begin{scope}[shift={(15,0)}]
            \node[empty] (1) {};
            \node[empty] (2) [right=of 1,shift={(1,0)}] {};
            \node[empty] (3) [right=of 2,shift={(1,0)}] {};
            \node[empty] (4) [right=of 3,shift={(1,0)}] {};
            \node at ($(1)+(-.1,-.6)$) {$v_1$};
            \node at ($(2)+(-.1,-.6)$) {$v_2$};
            \node at ($(3)+(-.1,-.6)$) {$v_{k-1}$};
            \node at ($(4)+(-.1,-.6)$) {$v_k$};
            \draw[-Stealth] (1) to node[below] {$e_1$} node[above] {\scriptsize $a(e_1)+a(e_k)$} (2);
            \draw[dashed] (2) to (3);
            \draw[-Stealth] (3) to node[below] {$e_{k-1}$} node[above] {\scriptsize $a(e_{k-1}) + a(e_k)$} (4);

            \foreach \x in {1,2,3,4}
                \draw [fill] (\x) circle [radius=.2];
        \end{scope}

        \end{tikzpicture}

        \caption{If we have a cycle where all edges have the same label, we can delete
            one of the edges and adjust $a(e_i)$ for the remaining edges $e_i$ to get a picture
            with the same phase.}\label{fig:cycledel}
    \end{figure}
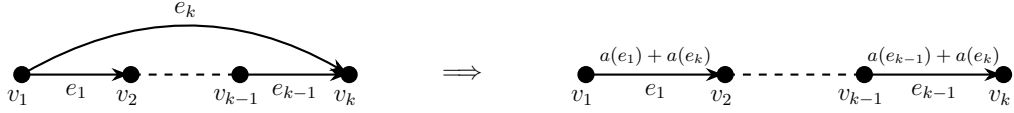
    Finally, for part (d), suppose $C = (v_1,e_1,\ldots,v_k,e_k,v_1)$ is a 
    simple cycle in $\mcP$ with $h(e_i)=h(e_j)$ for all $i,j$. By Lemma
    \ref{lem:orientation}, we can assume that $s(e_i) = v_i$ for all $1 \leq i
    \leq k-1$, and $s(e_k)=v_1$, as shown on the left in \Cref{fig:cycledel}.
    Deleting $e_k$ and changing the labels of $e_i$ from $a(e_i)$ to $a(e_i) +
    a(e_k)$ for all $1 \leq i \leq k-1$ gives a new picture with the same phase
    as $\mcP$ and one less edge, as shown on the right in \Cref{fig:cycledel}.
\end{proof}

We need some simple lemmas about hypergraphs and planar graphs:
\begin{lemma}\label{lem:girth}
    The girth of a hypergraph is the smallest $k$ such that there is a sequence
    $(v_1,e_1,\ldots,v_k,e_k,v_{k+1})$ where $v_1,\ldots,v_k$ is a sequence
    of distinct vertices, $v_{k+1} = v_1$, $e_1,\ldots,e_k$ is a sequence of
    edges containing at least two different edges, and $e_i$ is incident to
    $v_i$ and $v_{i+1}$ for all $1 \leq i \leq k$. 
\end{lemma}
\begin{proof}
    A Berge sequence satisfies the conditions of the lemma, so it suffices to
    show that if $(v_1,e_1,\ldots,e_k,v_{k+1})$ is a sequence satisfying the
    conditions of the lemma, then there is a Berge cycle of length $\leq k$.
    To prove this, first note that if $e_i = e_{i+1}$ for some $1 \leq i \leq
    k-1$, then $v_{i+2}$ is incident to $e_i$ and
    $e_1,\ldots,e_i,e_{i+2},\ldots,e_k$ still has at least two different edges,
    so $(v_1,\ldots,v_i,e_i,v_{i+2},e_{i+2},\ldots,e_k,v_{k+1})$ satisfies
    the conditions of the lemma. Thus we can assume that $e_i \neq e_{i+1}$ for
    all $1 \leq i \leq k-1$. If $e_i \neq e_j$ for all $1 \leq i \neq j \leq
    k$, then our sequence is a Berge cycle.  Otherwise, pick $1 \leq i_0 < j_0
    \leq k$ with $e_{i_0} = e_{j_0}$ minimizing $j_0-i_0$, so that
    $e_{i_0},\ldots,e_{j_0-1}$ is distinct.  Since $e_{i} \neq e_{i+1}$, we
    must have $j_0 > i_0+1$, so $e_i,\ldots,e_{j-1}$ contains at least two
    edges. Since $v_{j_0}$ is incident to $e_{j_0} = e_{i_0}$, we conclude that
    $(v_{j_0},e_{i_0},v_{i_0+1},\ldots,v_{j_0-1},e_{j_0-1},v_{j_0})$ is a Berge
    cycle of length $j_0-i_0 < k$.
\end{proof}

\begin{lemma}\label{lem:simplecycleinface}
    If $\mcP$ is a loopless plane graph with at least two faces, then for every
    face $F$, there is a simple cycle in $\mcP$ contained in the boundary of $F$.
\end{lemma}
\begin{proof}
    Every edge of $\mcP$ is contained in the boundary of either one or two
    faces.  Let $\mcP'$ be $\mcP$ with all the edges that are contained in the boundary
    of only one face deleted. Then $\mcP'$ has the same number of faces as
    $\mcP$, and every face $F$ of $\mcP$ is contained in a unique face $F'$ of $\mcP'$.
    Furthermore, the boundary of $F'$ is contained in the boundary of $F$. So
    we just need to show that every face $F'$ of $\mcP'$ contains a simple cycle.

    Pick an edge $e$ in the boundary of $F'$, and let
    $(v_1,e_1,v_2,e_2,\ldots,v_k,e_k)$ be the sequence of vertices and edges we
    encounter starting from $v_1 = s(e)$, and travelling along the boundary of
    $F'$ in the direction $e=e_1$ until we return back to the start of $e_1$.
    Because every edge of $\mcP'$ is contained in exactly two faces, the edges
    in the sequence $e_1,\ldots,e_k$ are distinct. Pick $1 \leq i < j \leq k+1$
    with $v_i = v_j$ minimizing $j-i$. Since $\mcP'$ is loopless, $j > i+1$,
    and since the vertices $v_i,\ldots,v_{j-1}$ are distinct,
    $(v_i,e_i,v_{i+1},e_{i+1},\ldots,v_{j-1},e_{j-1},v_i)$ is a simple cycle.
\end{proof}

If $F$ is a face of a plane graph $\mcP$, then the size of $F$ is $n_1 + 2
n_2$, where $n_1$ is the number of edges of $\mcP$ contained in the boundary of
$F$ and the boundary of some other face, and $n_2$ is the number of edges of
$\mcP$ contained in the boundary of $F$ and not contained in the boundary of
any other face. The following well-known lemma is the key idea of
combinatorial small cancellation theory. Since the proof is short, we include
it for the interested reader. 
\begin{lemma}\label{lem:smallcancellationgraph}
    Let $(a,b) \in \{(6,3),(4,4),(3,6)\}$. There are no plane graphs where
    every vertex has degree $\geq a$ and every face has size $\geq b$.
\end{lemma}
\begin{proof}
    Suppose $\mcP$ is a plane graph with $V$ vertices, $E$ edges, and $F$
    faces. If every vertex has degree $\geq a$, then the number of edges $E
    \geq a V / 2$. Similarly, if every face has size $\geq b$, then $E \geq b F / 2$.
    If $(a,b) \in \{(6,3),(4,4),(3,6)\}$, then $2/a + 2/b = 1$, and $E \geq V +
    F$. But this contradicts Euler's formula that $V + F - E = 2$.
\end{proof}

\begin{proof}[Proof of \Cref{T:main}]
    Suppose $\mcP$ is a minimal picture with phase $k$. By part (a) of Lemma \ref{lem:minimal},
    $\mcP$ is loopless, and by part (c), every vertex $v$ of $\mcP$ has degree $\geq d$. If 
    $\mcP$ has at least two faces, then Lemma \ref{lem:simplecycleinface} implies that every
    face $F$ of $\mcP$ has a simple cycle $(v_1,e_1,\ldots,v_m,e_m,v_{m+1})$ in its boundary.
    By part (b) of \Cref{lem:minimal}, the vertices $h(v_1),\ldots,h(v_m)$ are distinct,
    and by part (d), the sequence $h(e_1),\ldots,h(e_m)$ contains at least two different edges. 
    By \Cref{lem:girth}, we conclude that $m \geq g$, and hence every face of $\mcP$ has
    size $\geq g$, contradicting \Cref{lem:smallcancellationgraph}. 

    The other possibility is that $\mcP$ has just one face. In this case,
    $\mcP$ must either have no vertices and no edges, or be a forest. But since
    every vertex of $\mcP$ has degree $\geq 3$, $\mcP$ must have no vertices and
    edges, and thus $k=0$. By the van
    Kampen lemma, if $J^k=1$ in $\Gamma_p(A,b)$ then $k=0$ in $\Z_p$.
\end{proof}
Because we can't exclude the possibility of vertices of degree $4$ in part (b)
of \Cref{lem:minimal}, this proof doesn't work if $(d,g) = (6,3)$. However, we
do not know of a counterexample to the conclusion of \Cref{T:main} when
$(d,g)=(6,3)$. 

As mentioned in the introduction, a number of natural questions are not
answered by our work, particularly:
\begin{problem}\label{prob:finite}
    Is there $2 < p \leq +\infty$ and a linear system $Ax=b$ over $\Z$ with
    no classical solution, such that $\Gamma_p(A,b)$ has a finite-dimensional
    representation $\phi : \Gamma_p(A,b) \to U(\C^d)$ with $|\phi(J)| = p$? 
\end{problem}

\begin{problem}
    Is there $2 < p \leq +\infty$ and a linear system $Ax=b$ over $\Z$ such
    that $|J|=p$ in $\Gamma_p(A,b)$, but there is no finite-dimensional
    representation $\phi : \Gamma_p(A,b) \to U(\C^d)$ with $|\phi(J)| = p$? 
\end{problem}
For both problems, we are especially interested in examples with $p$ prime. A
solution to \Cref{prob:finite} has been announced by van Dobben de Bruyn
\cite{vDdB}.


\section{Incidence systems of graphs}\label{sec:incidence}

In this section, we look a bit closer at the incidence systems in
\Cref{ex:incidence}. As in \Cref{ex:K33}, we allow graphs and digraphs $G$
with arbitrary (finite) vertex and edge sets, which we denote by $V(G)$ and
$E(G)$ respectively. With this convention, the incidence matrix $I(G)$ of a 
digraph $G$ is an element of $\Z^{V \times E}$, and we get an incidence system
$I(G)x=b$ for any vector $b \in \Z^V$. Furthermore, the solution group
$\Gamma_p(I(G),b)$ is generated by $J$ and $x_e, e \in E$. The relations in
part (2) of \Cref{def:solngroup} state that 
\begin{equation*}
    \prod_{t(e)=v} x_e \cdot \prod_{s(e)=v} x_{e}^{-1} = J^{b_v}
\end{equation*}
for all vertices $v \in V$, and the relations in part (3) state that $x_e$ and
$x_{e'}$ commute if the edges $e$ and $e'$ are incident to a common vertex. 
A vector $b : V \to \Z$ is called a \textbf{$\Z$-weighting} of $G$, and a pair $(G,b)$
is called a \textbf{$\Z$-weighted graph}.

Given two linear systems $Ax=b$ and $A'x'=b'$, we say that a homomorphism
$\Gamma_p(A,b) \to \Gamma_p(A',b')$ is a \textbf{homomorphism over $\Z_p$} if
it sends $J \mapsto J$. Similarly, an \textbf{isomorphism over $\Z_p$} is a
homomorphism over $\Z_p$ which is a group isomorphism. Our starting point
for this section is the following lemma:
\begin{lemma}\label{lem:changeorientation}
    Suppose $(G,b)$ is a $\Z$-weighted graph, and $D$ and $D'$ are digraphs
    with underlying undirected graph $G$. For every $2 \leq p \leq +\infty$,
    there is an isomorphism $\Gamma_p(I(D),b) \to \Gamma_p(I(D'),b)$ over $\Z_p$.
\end{lemma}
\begin{proof}
    Both solution groups are generated by $J$ and $x_e$, $e \in E(G)$. 
    Looking at the relations in \Cref{def:solngroup}, we
    see that there is an isomorphism $\Gamma_p(I(D),b) \to \Gamma_p(I(D'),b)$
    which sends $J \mapsto J$ and
    \begin{equation*}
        x_e \mapsto \begin{cases} x_e & e \text{ has the same orientation in } D \text{ and } D' \\
                        x_e^{-1} & e \text{ has the opposite orientation in } D \text{ and } D' \\
                \end{cases}\ .
    \end{equation*}
\end{proof}
As a result, if $G$ is an undirected graph, we'll use $I(G)$ to refer to the
incidence matrix of $G$ with some arbitrarily chosen edge orientations, and
think of $\Gamma_p(I(G),b)$ as being associated with $G$ as an undirected
graph.  If $G$ is connected (as an undirected graph), then the
solution group does not depend on $b$ either, just on $|b| := \sum_v b_v$.
\begin{lemma}\label{lem:changeb}
    Let $G$ be a connected graph, and $2 \leq p \leq +\infty$. If $b, b' : V(G) \to \Z$
    are two $\Z$-weightings such that $|b| = |b'|$ in $\Z_p$, then $\Gamma_p(I(G),b)$
    is isomorphic to $\Gamma_p(I(G),b')$ over $\Z_p$.
\end{lemma}
\begin{proof}
    Fix some orientation for the edges of $G$. If $e \in E$, let $b(e) \in \Z^V$
    be the vector such that $b(e)_{s(e)} = -1$, $b(e)_{t(e)} = 1$, and $b(e)_v = 0$
    otherwise. For any $\lambda \in \Z$, there is an isomorphism
    \begin{equation*}
        \Gamma_p(I(G),b) \to \Gamma_p(I(G),b+\lambda b(e)) : J \mapsto J, x_f \mapsto \begin{cases}
                J^{\lambda} x_f & f = e \\
            x_f & f \neq e \end{cases}.
    \end{equation*}
    $G$ is connected if and only if the vectors $b(e)$ span the subgroup $\{c
    \in \Z^V : |c|=0\}$. If $|b| = |b'|$ in $\Z_p$ then there is $c \in \Z^V$
    with $|c|=0$ such that $b+c = b'$ in $\Z_p$. By writing $c$ as a linear
    combination of vectors $b(e)$ and successively applying the isomorphisms
    above, we can get an isomorphism $\Gamma_p(I(G),b) \to \Gamma_p(I(G),b+c)
    = \Gamma_p(I(G),b')$. 
\end{proof}
If $(G,b)$ is a disconnected $\Z$-weighted graph with components $G_i$,
and $b_i = b|_{V(G_i)}$, then
\begin{equation*}
    \Gamma_p(I(G),b) \iso \ast_J\ \Gamma_p(I(G_i), b_i), 
\end{equation*}
the group constructed by taking the free product of the groups
$\Gamma_p(I(G_i),b_i)$ and then identifying the elements $J$ in each
factor. As a result, we can restrict our investigation to connected graphs.

\Cref{lem:changeb} can be combined with the following proposition for general
linear systems. In stating the proposition, we use standard arithmetic with the
extended natural numbers, so for instance any positive integer divides $+\infty$.
\begin{prop}\label{prop:changeb}
    Let $Ax=b$ be an $m \times n$ linear system over $\Z$.
    \begin{enumerate}[(a)]
        \item If $2 \leq p,q \leq +\infty$ and $\alpha,\beta,\lambda,\delta$ are integers such that
            $q$ divides $p \gcd(\lambda,\delta)$ and $\delta \alpha = \beta \lambda$ in $\Z_q$,
            then there is a homomorphism
            \begin{equation*}
                \Gamma_{p}(A,\alpha b) \to \Gamma_q(A,\beta b)
            \end{equation*}
            sending $x_j \mapsto x_j^{\lambda}$ for all $j \in [n]$ and $J \mapsto J^{\delta}$.
        \item Suppose $2 \leq r,s < +\infty$ are coprime. If $|J|=k$ in $\Gamma_r(A,b)$
            and $|J|=l$ in $\Gamma_s(A,b)$, then $|J|=kl$ in $\Gamma_{rs}(A,b)$. 
        \item If $2 \leq p \leq +\infty$ and $t$ is a non-zero divisor in $\Z_p$, then
            $|J|=p$ in $\Gamma_p(A,b)$ if and only if $|J|=p$ in $\Gamma_p(A,t b)$. If 
            $p$ is finite, then these two groups are isomorphic over $\Z_p$.
        \item Suppose $2 \leq p < +\infty$ and $t$ is a positive integer. If $|J|=p$ in $\Gamma_p(A,b)$, then 
            $|J|=tp$ in $\Gamma_{tp}(A,tb)$.
        \item If $b=0$, then $|J|=p$ in $\Gamma_p(A,b)$ for all $2 \leq p \leq +\infty$.
    \end{enumerate}
\end{prop}
Similar results to parts (a) and (b) have previously appeared in \cite[Section 6.3]{OK20}
and \cite[Corollary 3.21]{COS}.
\begin{proof}
    For part (a), $J^{\delta}$ commutes with $x_j^{\lambda}$ in
    $\Gamma_q(A,\beta b)$ for all $j \in [n]$, and if $A_{ij}, A_{ik} \neq 0$ for
    some $i \in [m]$, $j,k \in [n]$, then $x_j^{\lambda}$ and $x_k^{\lambda}$
    commute as well. For all $i \in [m]$, we have
    \begin{equation*}
        \prod_j (x_j^{\lambda})^{A_{ij}} = J^{\lambda \beta b_i} = (J^{\delta})^{\alpha b_i}
    \end{equation*}
    in $\Gamma_q(A,b)$. Finally, if $p < +\infty$, then $q < +\infty$, and 
    $J^{p \delta} = x_j^{p \lambda} = 1$ in $\Gamma_q(A,b)$ since $q$ divides
    $p\delta$ and $p \lambda$. It follows that there is a homomorphism as required. 

    For part (b), let $|J|=a$ in $\Gamma_{rs}(A,b)$. Applying part (a) with
    $p=rs$, $q=r$, and $\alpha=\beta=\delta=\lambda=1$ gives a homomorphism
    $\Gamma_{rs}(A,b) \to \Gamma_r(A,b)$ sending $J \mapsto J$, and similarly
    there is a homomorphism $\Gamma_{rs}(A,b) \to \Gamma_{s}(A,b)$ sending $J
    \mapsto J$. Thus there is a homomorphism $\Gamma_{rs}(A,b) \to
    \Gamma_{r}(A,b) \times \Gamma_s(A,b)$ sending $J \mapsto (J,J)$. Since $k$
    and $l$ are coprime (as they divide $r$ and $s$ respectively), we conclude
    that the order of $(J,J)$ in the product group is $kl$, so $kl$ divides
    $a$. Applying part (a) with $p=r$, $q=rs$, $\alpha=\beta=1$, and $\delta=\lambda=s$,
    we see that there is a homomorphism $\Gamma_r(A,b) \to \Gamma_{rs}(A,b)$ sending
    $J \mapsto J^s$, from which we conclude that $a$ divides $ks$. Similarly, $a$
    divides $rl$, so $a$ divides $\gcd(ks,rl)=kl$, and hence $a=kl$.

    For part (c), applying part (a) with $q=p$, $\alpha = \lambda = t$, and
    $\delta = \beta = 1$ gives a homomorphism $\phi : \Gamma_p(A,tb) \to
    \Gamma_p(A,b)$ sending $x_i \mapsto x_i^{t}$ and $J \mapsto J$. If $p < +\infty$,
    then $t$ is coprime to $p$, and hence there is an integer $s$ such that $st=1$
    in $\Z_p$. Replacing $t$ with $s$ and $b$ with $tb$, we get a homomorphism
    $\Gamma_p(A,b) = \Gamma_p(A,st b) \to \Gamma_p(A,tb)$ sending $x_i \mapsto x_i^s$
    and $J \mapsto J$, and this is an inverse to $\phi$. If $p=+\infty$, then 
    $t$ does not have to be invertible in $\Z_p = \Z$, but we can still conclude
    from $\phi$ that if $|J|=+\infty$ in $\Gamma_p(A,b)$, then $|J|=+\infty$ in
    $\Gamma_p(A,tb)$. For the reverse, applying part (a) with $q=p$,
    $\alpha=\lambda=1$, and $\delta=\beta=t$ gives a homomorphism $\phi :
    \Gamma_p(A,b) \to \Gamma_p(A,tb)$ sending $J \mapsto J^t$. Hence if
    $|J|=+\infty$ in $\Gamma_p(A,tb)$, then $|J|=+\infty$ in $\Gamma_p(A,b)$.

    For part (d), let $K$ be the finitely presented group constructed from the
    presentation of $\Gamma_p(A,b)$ by adding a new generator $J_0$ which commutes
    with all other generators, along with the relation $J = (J_0)^t$. Since
    $J_0^{tp}=1$ in $K$, the group $K$ is the central product of $\Gamma_p(A,b)$
    with $\Z_{pt} = \ang{J_0}$ identifying the subgroups $\ang{J} = \Z_p = 
    \ang{J_0^t}$. As a result, $\Gamma_p(A,b)$ and $\Z_{pt}$ are both subgroups
    of $K$, and in particular $|J_0| = pt$. Checking relations as in part (a),
    we see that there is a homomorphism $\Gamma_{tp}(A,tb) \to K$ sending
    $x_i \mapsto x_i$ and $J \mapsto J_0$, so $J$ has order $tp$ in $\Gamma_{tp}(A,tb)$. 

    Part (e) follows immediately from \Cref{def:solngroup}, since $J$ does not occur
    in any relations when $b=0$.
\end{proof}

Combining \Cref{prop:changeb} with \Cref{lem:changeb}, we immediately get:
\begin{cor}\label{cor:changeb}
    Let $(G,b)$ be a $\Z$-weighted connected graph with $|b|=k \in \Z_p$, and
    suppose $2 \leq p \leq +\infty$. 
    \begin{enumerate}[(a)]
        \item If $p=rs$ where $r$ and $s$ are coprime, then $|J|=p$ in $\Gamma_p(I(G),b)$
            if and only if $|J|=r$ in $\Gamma_r(I(G),b)$ and $|J|=s$ in $\Gamma_s(I(G),b)$.  
        \item Let $\ell$ be the order of $k$ in $\Z_p$ (for instance, if $k$ divides $p$ then
            $\ell = p/k$). Then $|J|=p$ in $\Gamma_p(I(G),b)$ if and only if
            $|J|=p$ in $\Gamma_p(I(G),b')$ for all other $\Z$-weightings $b'$
            such that $|b'|$ has order $\ell$ in $\Z_p$.
        \item If $k=0$, then $|J|=p$ in $\Gamma_p(I(G),b)$. 
    \end{enumerate}
\end{cor}
By part (a) of \Cref{cor:changeb}, the question of whether $|J|=p$ in
$\Gamma_p(I(G),b)$ can be reduced to the case that $p$ is $+\infty$ or a prime
power. Part (b) shows that whether $|J|=p$ in $\Gamma_p(I(G),b)$ depends only
on the order of $|b|$ in $\Z_p$. If $p$ is $+\infty$ or prime, then every
non-zero element of $\Z_p$ has order $p$, so $|J|=p$ in $\Gamma_p(I(G),b)$ for
some $\Z$-weighting $b$ with $|b| \neq 0$ in $\Z_p$ if and only if $|J|=p$ in
$\Gamma_p(I(G),b')$ for all $\Z$-weightings $b'$ with $|b'| \neq 0$ in $\Z_p$.
If $p$ is a prime power, i.e. $p=q^a$ for some prime $q$ and $a \geq 1$, then
the set of possible orders of non-zero elements of $\Z_p$ is $\{q^i : 1 \leq i
\leq a\}$. As the following example shows, whether or not $|J|=p$ in $\Gamma_p(I(G),b)$
can depend on the order of $|b|$, and not just on whether $|b|$ is non-zero.
\begin{example}
    Let $b$ be a $\Z$-weighting of $K_{3,3}$ with $|b|=1$. By \Cref{ex:K33}, $|J|<4$
    in $\Gamma_4(I(K_{3,3}),b)$. However, since $|J|=2$ in $\Gamma_2(I(K_{3,3}),b)$,
    it follows from part (d) of \Cref{prop:changeb} that $|J|=4$ in 
    $\Gamma_4(I(K_{3,3}),2b)$. 
\end{example}

Recall that a graph $G'$ is an algebraic graph minor of a graph $G$ if it is
possible to construct $G'$ from $G$ via a sequence of edge deletions, edge
contractions, and vertex deletions. A graph $G$ avoids a set
of graphs $\mcF$ if there is no $G' \in \mcF$ which is a minor of $G$.  A
property $(P)$ of (connected) graphs $G$ is said to be minor-closed if
whenever a (connected) graph $G'$ is a minor of a connected graph $G$ satisfying
$(P)$, then $G'$ also satisfies $(P)$. The following is a consequence of the
Robertson-Seymour theorem:
\begin{theorem}[\cite{RS04}]\label{thm:robertsonseymour}
    Let $(P)$ be a minor-closed property of connected graphs. Then there is a
    finite set of connected graphs $\mcF$ such that a connected graph $G$ has
    $(P)$ if and only if $G$ avoids $\mcF$. 
\end{theorem}
A set $\mcF$ characterizing $(P)$ in this way is called a set of forbidden
minors for $(P)$. The following lemma can be used to show that many properties
of $G$ arising from $\Gamma_p((I(G),b)$ are minor-closed. 
\begin{lemma}[\cite{PRSS}]\label{lem:minorclosed}
    Suppose $(G,b)$ and $(G',b')$ are $\Z$-weighted connected graphs, and $2
    \leq p \leq +\infty$. If $G'$ is a minor of $G$ and $|b'| = |b|$ in $\Z_p$,
    then there is a surjective group homomorphism $\Gamma_p(I(G),b) \to
    \Gamma_p(I(G'),b')$ over $\Z_p$.
\end{lemma}
\begin{proof}
    The proof in \cite{PRSS} is for $p=2$, but is easily adapted to
    the general case. 
\end{proof}

\begin{cor}\label{cor:minors}
    Let $2 \leq p \leq +\infty$ and $\ell$ be the order of an element of $\Z_p$. 
    Then there is a finite set of connected graphs $\mcF$ such that if $(G,b)$
    is a connected $\Z$-weighted
    graph where $|b|$ has order $\ell$ in $\Z_p$, then $|J|=p$ in
    $\Gamma_p(I(G),b)$ if and only if one of the graphs in $\mcF$ is a minor of
    $G$. 
\end{cor}
\begin{proof}
    Let $(P)$ be the property that $|J|<p$ in $\Gamma_p(I(G),b)$ for some vertex
    weighting $b : V(G) \to \Z$ where $|b|$ has order $\ell$ in $\Z_p$. By Lemma \ref{lem:changeb}, this
    is equivalent to the property that $|J|<p$ in $\Gamma_p(I(G),b)$ for all
    vertex weightings $b : V(G) \to \Z$ where $|b|$ has order $\ell$ in $\Z_p$. Suppose $(P)$ holds
    for a graph $G$, and let $b$ be a vertex weighting of $G$ such that
    $|b|$ has order $\ell$ in $\Z_p$. 
    If $G'$ is a connected minor of $G$ and $b'$ is a vertex weighting of $G'$
    with $|b'|=|b|$, then by \Cref{lem:minorclosed}, there is a surjective
    group homomorphism $\Gamma_p(I(G),b) \to \Gamma_p(I(G'),b')$ over $\Z_p$.
    It follows that $|J|<p$ in $\Gamma_p(I(G'),b')$, so $G'$ has property $(P)$. 
    Hence $(P)$ is minor-closed, and by \Cref{thm:robertsonseymour}, there is a
    finite list of forbidden minors characterizing $(P)$. 
\end{proof}
Per \Cref{cor:changeb}, if $|b|=0$ then we always have $|J|=p$ in $\Gamma_p(I(G),b)$,
and hence we can take $\mcF$ in \Cref{cor:minors} to consist of the graph
with no vertices and no edges. If $p=\ell=2$, then by a theorem of Arkhipov
mentioned in the introduction, we can take $\mcF$ to be the minors $K_{3,3}$ and
$K_5$ for planarity. The other cases are open:
\begin{problem}\label{prob:forbidden}
    Given $2 < p \leq +\infty$ and an order $\ell \neq 1$ of an element in $\Z_p$,
    find the forbidden minors $\mcF$ in \Cref{cor:minors}.
\end{problem}
In the case that $p$ is prime or $+\infty$, all elements have order $1$ or $p$.
Suppose that we restrict to simple graphs, i.e. graphs
with no loops, and at most one edge between any two vertices. The neighbour set
$N_G(v)$ of a vertex $v$ in a simple graph $G$ is the set of vertices $w$
which are adjacent to $v$.  A graph homomorphism $H \to G$ between
simple graphs $H$, $G$ is a function $\phi : V(H) \to V(G)$ such that if $v$
and $w$ are adjacent in $H$, then $\phi(v)$ and $\phi(w)$ are adjacent in $G$. 
A graph homomorphism is a cover if $\phi$ is surjective, and $\phi|_{N_H(v)}$
is a bijection between $N_H(v)$ and $N_G(\phi(v))$ for all $v \in V(H)$.  If
$G$ is connected and $\phi : H \to G$ is a cover, then there is a natural
number $k$, called the arity of $\phi$,  such that $|\phi^{-1}(v)|=k$
for all vertices $v \in V(G')$. A $k$-cover is a cover of arity $k$. A cover $H
\to G$ is planar if $H$ is a planar graph.

\begin{lemma}\label{lem:cover}
    Let $(G,b)$ be a $\Z$-weighted simple connected graph, and suppose $G$
    has a planar $k$-cover. Then $J^{k|b|} = 1$ in $\Gamma_p(I(G),b)$ for all
    $2 \leq p \leq +\infty$.
\end{lemma}
\begin{proof}
    The proof is exactly the same as \Cref{ex:K33}. Suppose we have chosen some
    orientation for the edges of $G$, so that $I(G)$ is defined. Let $\phi : H \to
    G$ be a planar $k$-cover of $G$. We can make $H$ into a picture for $\Gamma_p(I(G),b)$.
    Indeed, set $h_V = \phi$, and if $e \in E(H)$ has endpoints $v$ and $w$, let 
    $h_E(e)$ be the edge of $G$ with endpoints $\phi(v)$ and $\phi(w)$. Pick a plane
    embedding of $H$, and orient the edges of $H$ to match the edges of $G$, or
    in other words so that if $e \in E(H)$, then $s(h(e)) = h(s(e))$ and $t(h(e))
    = h(t(e))$. Finally set $a(e)=1$ for all $e \in E(H)$, and $k(v) = b_{\phi(v)}$
    for all $v \in V(H)$. Then $H$ is a picture with phase $k |b|$.
\end{proof}
If $|b|=1$ and $G$ has a planar $k$-cover, then the lemma implies that $J^k=1$
in $\Gamma_p(I(G),b)$. If $G$ is planar, then $G$ is a planar $1$-cover of
itself, and $J=1$ in $\Gamma_p(I(G),b)$. As we saw in \Cref{ex:K33} for
$K_{3,3}$, if $G$ has a planar double cover, then $J^2=1$ in $\Gamma_2(I(G),b)$,
and it is not possible to have $|J|=p$ in $\Gamma_p(I(G),b)$ unless $p=2$.  A
theorem of Negami states that a connected graph $G$ has a planar double cover
if and only if $G$ is projective planar \cite{Neg86}. No example of a graph
with a planar $k$-cover but no planar $2$-cover is known. Negami's
$1$--$2$--$\infty$ conjecture states that a connected graph $G$ has a planar
cover if and only if it is planar or projective planar, or equivalently if and
only if it has a planar $1$-cover or $2$-cover (\cite{Neg88}, see also the
survey \cite{Hli08}). The forbidden minors for projective planarity are known
\cite{GHW79,Arch81}, so if true, Negami's conjecture would determine the
forbidden minors for having a planar cover. 
We might wonder if the converse of \Cref{lem:cover} is true: does every
picture showing $J^k=1$ in $\Gamma_p(I(G),b)$ come from a planar $k$-cover of
$G$? If so, then Negami's conjecture would also imply that the minors
for \Cref{prob:forbidden} with $2 < p \leq +\infty$ and $\ell=p$ are just the
forbidden minors for projective planarity. Unfortunately, this isn't the
case, as the following example shows.

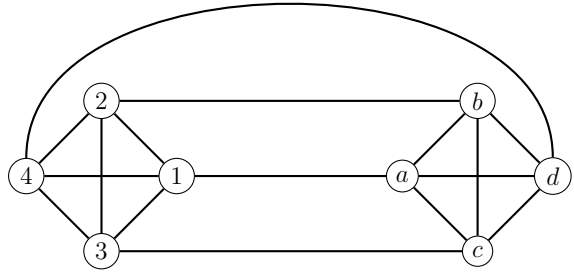
\begin{figure}
\begin{tikzpicture}[auto, thick, scale=.5,vertex/.style={circle,draw,thin,inner sep=2.5}, empty/.style={inner sep=0}, every node/.style={scale=.8}]
    \node[vertex] (1) at (2,0) {$1$};
    \node[vertex] (2) at (0,2) {$2$}
        edge (1);
    \node[vertex] (3) at (0,-2) {$3$}
        edge (1) edge (2);
    \node[vertex] (4) at (-2,0) {$4$}
        edge (1) edge (2) edge (3);

    \node[vertex] (a) at (8,0) {$a$}
        edge (1);
    \node[vertex] (b) at ($(a) + (2,2)$) {$b$}
        edge (a) edge (2);
    \node[vertex] (c) at ($(a) + (2,-2)$) {$c$}
        edge (a) edge (b) edge (3);
    \node[vertex] (d) at ($(a) + (4,0)$) {$d$}
        edge (a) edge (b) edge (c);

    \draw (d) to [out=90,in=90] (4);
\end{tikzpicture}
\caption{The graph $D_{17}$, which does not have a planar $2$-cover.}\label{fig:D17}
\end{figure}

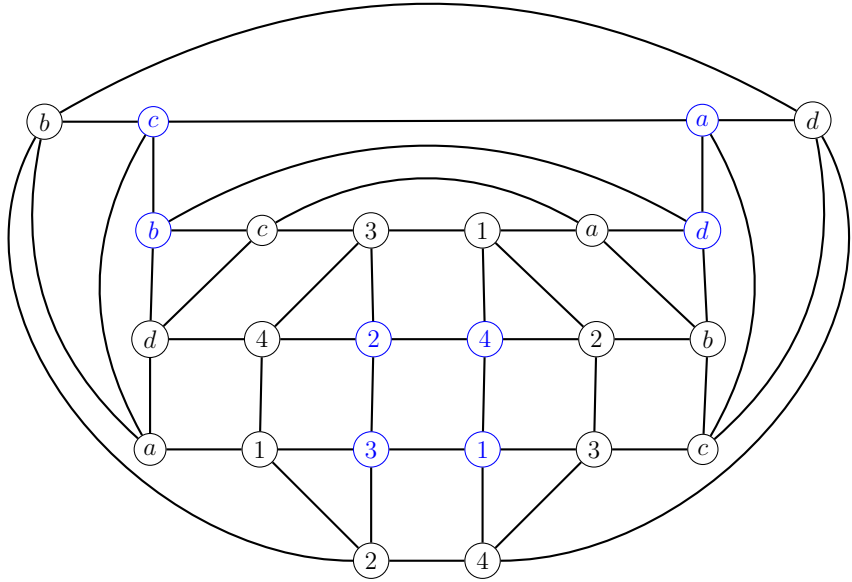
\begin{figure}
\begin{tikzpicture}[auto, thick, scale=.5,vertex/.style={circle,draw,thin,inner sep=2.5}, empty/.style={inner sep=0}, every node/.style={scale=.8}]

    \node[vertex] (c1) {$c$};
    \node[vertex] (31) [right=of c1] {$3$}
        edge (c1);
    \node[vertex] (11) [right=of 31] {$1$}
        edge (31);
    \node[vertex] (a1) [right=of 11] {$a$}
        edge (11) edge [bend right] (c1);

    \node[vertex] (41) [below=of c1] {$4$}
        edge (31);
    \node[vertex,color=blue] (21) [right=of 41] {$2$}
        edge (41) edge (31);
    \node[vertex,color=blue] (42) [right=of 21] {$4$}
        edge (21) edge (11);
    \node[vertex] (22) [right=of 42] {$2$}
        edge (42) edge (11);
    \node[vertex] (d1) [left=of 41] {$d$}
        edge (41) edge (c1);
    \node[vertex] (b1) [right=of 22] {$b$}
        edge (a1) edge (22);

    \node[vertex] (a2) [below=of d1] {$a$}
        edge (d1);
    \node[vertex] (12) [right=of a2] {$1$}
        edge (a2) edge (41);
    \node[vertex,color=blue] (32) [right=of 12] {$3$}
        edge (12) edge (21);
    \node[vertex,color=blue] (13) [right=of 32] {$1$}
        edge (32) edge (42);
    \node[vertex] (33) [right=of 13] {$3$}
        edge (13) edge (22);
    \node[vertex] (c2) [right=of 33] {$c$}
        edge (33) edge (b1);

    \node[vertex] (23) [below=of 32] {$2$}
        edge (12) edge (32);
    \node[vertex] (43) [right=of 23] {$4$}
        edge (13) edge (33) edge (23);

    \node[vertex,color=blue] (b2) [left=of c1] {$b$}
        edge (c1) edge (d1);
    \node[vertex,color=blue] (d2) [right=of a1] {$d$}
        edge (a1) edge (b1) edge [bend right] (b2);

    \node[vertex,color=blue] (c3) [above=of b2] {$c$}
        edge [bend right] (a2) edge (b2);
    \node[vertex,color=blue] (a3) [above=of d2] {$a$}
        edge (d2) edge [bend left] (c2) edge (c3);

    \node[vertex] (b3) [left=of c3] {$b$}
        edge [bend right] (a2) edge [out=-120,in=180] (23) edge (c3);
    \node[vertex] (d3) [right=of a3] {$d$}
        edge [bend right] (b3) edge (a3) edge [bend left] (c2) edge [out=-60,in=0] (43);
\end{tikzpicture}
\caption{A picture showing that $J^{2|b|}=1$ in $\Gamma_p(I(D_{17}),b)$.}\label{fig:D17pic}
\end{figure}

\begin{example}\label{ex:D17}
    The graph $D_{17}$, shown in \Cref{fig:D17} with vertex set $W=\{1,2,3,4,a,b,c,d\}$,
     is one of the forbidden minors for projective planarity
    \cite{GHW79}. Thus $D_{17}$ does not have a planar $2$-cover (in fact,
    it's known that $D_{17}$ does not have a  planar cover of any arity
    \cite{Neg88b}). 
    A picture $\mcP$ for $\Gamma_p(I(D_{17}),b)$ is shown in \Cref{fig:D17pic}. The
    vertices of $\mcP$ are labelled by elements of $W$ indicating the
    function $h_V$ of the picture. The function $h_E$ is described implicitly
    by the vertex labels, so if $e$ is an edge of $\mcP$ with endpoints $v$ and
    $w$, then $h_E(e)$ is mapped to the edge of $D_{17}$ with endpoints
    $h_V(v)$ and $h_V(w)$. Edge orientations are not shown in
    \Cref{fig:D17pic}, but should be chosen to match the edge orientations we
    pick for $D_{17}$. With this choice, we can take $a(e)=1$ for all edges $e$.
    Finally, if $v$ is one of the vertices indicated in blue, we set $k(v)=0$.
    Otherwise we $k(v) = b(h(v))$. With these labellings, $\mcP$ is a picture
    for $\Gamma_p(I(D_{17}),b)$ with phase $2|b|$. 
    If $|b|=1$, then this picture shows that $J^2=1$ in $\Gamma_p(I(D_{17}),b)$ for
    all $2 \leq p \leq +\infty$, even though $D_{17}$ does not have a planar
    $2$-cover. 
    
    Note that if $v$ is a vertex of $\mcP$, then $h(N_\mcP(v)) \subseteq
    N_{D_{17}}(h(v))$, so $h_V$ is a graph homomorphism from $\mcP$ to $D_{17}$. 
    If $v$ is one of the vertices indicated in black, then $h(N_{\mcP(v)}) 
    = N_{D_{17}}(h(v))$, so $h_V$ looks like a cover at these vertices. However,
    if $v$ is one of the vertices indicated in blue, then $h(N_{\mcP(v)})$
    contains only two vertices of $N_{D_{17}}(h(v))$, so $h_V$ fails to be a
    cover at these vertices. 
\end{example}

Although the answer to \Cref{prob:forbidden} for $p>2$ does not exactly line
up with the non-projective planar graphs, it is interesting that that for all
the examples we know so far, if $|J|=p$ in $\Gamma_p(I(G),b)$ for $|b|=1$ and
some $p > 2$, then $|J|=p$ in $\Gamma_p(I(G),b)$ for all $2 \leq p \leq
+\infty$. This raises the question of whether a version of the $1$--$2$--$\infty$
conjecture holds for solution groups:
\begin{problem}\label{prob:negami}
    Is there a $\Z$-weighted connected graph $(G,b)$ with $|b|=1$ and
    $2 < p,q \leq +\infty$ such that
    such that $|J|=p$ in $\Gamma_p(I(G),b)$ but $|J| < q$ in $\Gamma_q(I(G),b)$?
\end{problem}
A similar question has been asked independently by van Dobben de Bruyn and
Roberson \cite{vDdBR}.  There are other interesting variations of
\Cref{prob:forbidden}. For instance, if the answer to \Cref{prob:negami} is
negative, we can ask for the minors characterizing the property that $|J|=p$ in
$\Gamma_p(I(G),b)$ for all $2 \leq p \leq +\infty$.

\bibliographystyle{amsalpha}
\bibliography{small}

\providecommand{\bysame}{\leavevmode\hbox to3em{\hrulefill}\thinspace}
\providecommand{\MR}{\relax\ifhmode\unskip\space\fi MR }
\providecommand{\MRhref}[2]{%
  \href{http://www.ams.org/mathscinet-getitem?mr=#1}{#2}
}
\providecommand{\href}[2]{#2}
\begin{thebibliography}{GHW79}

\bibitem[Arc81]{Arch81}
Dan Archdeacon, \emph{A {K}uratowski theorem for the projective plane}, Journal
  of Graph Theory \textbf{5} (1981), no.~3, 243--246.

\bibitem[BH99]{BH}
Martin Bridson and Andre Haefliger, \emph{Metric {S}paces of {N}on-{P}ositive
  {C}urvature}, Grundlehren der mathematischen {W}issenschaften, Springer
  Berlin, Heidelberg, 1999.

\bibitem[CLS17]{CLS}
Richard Cleve, Li~Liu, and William Slofstra, \emph{Perfect commuting-operator
  strategies for linear system games}, Journal of Mathematical Physics
  \textbf{58} (2017), no.~01, 012202.

\bibitem[CM14]{CM14}
Richard Cleve and Rajat Mittal, \emph{Characterization of {Binary} {Constraint}
  {System} {Games}}, Automata, {Languages}, and {Programming}, Lecture {Notes}
  in {Computer} {Science}, no. 8572, Springer Berlin Heidelberg, 2014,
  arXiv:1209.2729, pp.~320--331.

\bibitem[COS24]{COS}
Ho~Yiu Chung, Cihan Okay, and Igor Sikora, \emph{Simplicial techniques for
  operator solutions of linear constraint systems}, Topology and its
  Applications \textbf{348} (2024), 108883.

\bibitem[{Dob}]{vDdB}
{Josse}~{van} {Dobben de Bruyn}, in preparation.

\bibitem[DR]{vDdBR}
{Josse}~{van} {Dobben de Bruyn} and David Roberson, \emph{Connections between
  solution groups, harmonic homomorphisms, and {N}egami's conjecture}, in
  preparation.

\bibitem[EL14]{EL14}
David Ellis and Nathan Linial, \emph{On {R}egular {H}ypergraphs of {H}igh
  {G}irth}, The Electronic Journal of Combinatorics \textbf{21} (2014), no.~1,
  P1.54.

\bibitem[FOC22]{FOC}
Markus Frembs, Cihan Okay, and Ho~Yiu Chung, \emph{No state-independent
  contextuality can be extracted from contextual measurement-based quantum
  computation with qudits of odd prime dimension}, preprint (2022),
  arXiv:2209.14018.

\bibitem[GHW79]{GHW79}
Henry~H. Glover, John~P. Huneke, and Chin~San Wang, \emph{103 graphs that are
  irreducible for the projective plane}, Journal of Combinatorial Theory,
  Series B \textbf{27} (1979), no.~3, 332--370.

\bibitem[Hli10]{Hli08}
Petr Hlin\v{e}n\'{y}, \emph{20 {Y}ears of {N}egami’s {P}lanar {C}over
  {C}onjecture}, Graphs and Combinatorics \textbf{26} (2010), no.~4, 525--536.

\bibitem[LS77]{LS77}
R.C. Lyndon and P.E. Schupp, \emph{Combinatorial {G}roup {T}heory}, Classics in
  Mathematics, Springer, 1977.

\bibitem[Mer90]{Me90}
N.~David Mermin, \emph{Simple unified form for the major no-hidden-variables
  theorems}, Physical Review Letters \textbf{65} (1990), no.~27, 3373--3376.

\bibitem[Neg86]{Neg86}
S.~Negami, \emph{Enumeration of {P}rojective-planar {E}mbeddings of {G}raphs},
  Discrete Mathematics \textbf{62} (1986), 299--306.

\bibitem[Neg88a]{Neg88b}
\bysame, \emph{Graphs which have no finite planar covering}, Bulletin of the
  Institute of Mathematics Academia Sinica \textbf{15} (1988), no.~4, 378--384.

\bibitem[Neg88b]{Neg88}
\bysame, \emph{The spherical genus and virtually planar graphs}, Discrete
  Mathematics \textbf{70} (1988), 159--168.

\bibitem[OR20]{OK20}
Cihan Okay and Robert Raussendorf, \emph{Homotopical approach to quantum
  contextuality}, Quantum \textbf{4} (2020), 217.

\bibitem[Per91]{Peres91}
Asher Peres, \emph{Two simple proofs of the {K}ochen-{S}pecker theorem},
  Journal of Physics A: Mathematical and General \textbf{24} (1991), no.~4,
  L175.

\bibitem[PRSS23]{PRSS}
Connor Paddock, Vincent Russo, Turner Silverthorne, and William Slofstra,
  \emph{Arkhipov{\textquoteright}s theorem, graph minors, and linear system
  nonlocal games}, Algebraic Combinatorics \textbf{6} (2023), no.~4,
  1119--1162.

\bibitem[QW20]{QW}
Hammam Qassim and Joel Wallman, \emph{Classical vs quantum satisfiability in
  linear constraint systems modulo an integer}, Journal of Physics A:
  Mathematical and Theoretical \textbf{53} (2020), 385304.

\bibitem[RS04]{RS04}
Neil Robertson and Paul~D. Seymour, \emph{Graph minors. {XX}. {W}agner's
  conjecture}, Journal of Combinatorial Theory, Series B \textbf{92} (2004),
  no.~2, 325--357.

\bibitem[Sho07]{Sh07}
Hamish Short, \emph{Diagrams and {G}roups}, The {G}eometry of the {W}ord
  {P}roblem for {F}initely {G}enerated {G}roups, Birkh{\"a}user Basel, Basel,
  2007.

\bibitem[Slo19]{Sl19}
William Slofstra, \emph{The set of quantum correlations is not closed}, Forum
  of Mathematics, Pi \textbf{7} (2019), E1.

\bibitem[Slo20]{Sl20}
\bysame, \emph{Tsirelson’s problem and an embedding theorem for groups
  arising from no n-local games}, Journal of the American Mathematical Society
  \textbf{33} (2020), 1--56.

\end{thebibliography}
\end{document}